\documentclass{article}
\usepackage[left=3.8cm, right=3.8cm, top=2.6cm]{geometry}

\usepackage{cite}
\usepackage{amsthm}
\usepackage{amsmath}
\usepackage{amssymb}
\usepackage{color}
\usepackage{algorithm}
\usepackage{algorithmic}

\usepackage{textcomp}
\usepackage[font=small, format=hang, labelfont={bf}]{caption}
\usepackage{graphicx}
\usepackage{epstopdf}
\usepackage{longtable}
\DeclareGraphicsExtensions{.eps}
\usepackage{float}
\usepackage{pgfplots}
\usepackage[thinspace,thickqspace,cdot,amssymb,textstyle]{SIunits}
\usepackage{bbm}
\usepackage{tikz}
\usetikzlibrary{arrows}
\usepackage{wrapfig}

\usepackage{wrapfig}
\usepackage{listings}
\usepackage{mathtools}

\usepackage{blkarray}

\usepackage{bbm}
\usepackage{amssymb,amsmath,dsfont}
\usepackage{color}
\usepackage{bm}

\newcommand{\real}{\mathbb{R}}

\newcommand{\integernonnegative}{\mathbb{Z}_{\ge 0}}

\newcommand{\R}{\mathbb{R}} 
\newcommand{\G}{\mathcal{G}} 
\newcommand{\V}{\mathcal{V}} 
\newcommand{\E}{\mathcal{E}} 

\newcommand{\ind}{\mathds{1}} 

\newcommand{\argmax}[1]{\underset{#1}{\mathrm{argmax\,}}}

\newcommand{\piu}{{{\boxplus}}}
\newcommand{\meno}{{{\boxminus}}}

\usepackage{flowchart}
\usetikzlibrary{arrows}
\usetikzlibrary{mindmap}

\title{Graph structure based Heuristics for Optimal Targeting in Social Networks}
\date{}
\author{Massimo Bini, Paolo Frasca, Chiara Ravazzi, Fabrizio Dabbene
\thanks{Massimo Bini is with Politecnico di Torino, Corso Duca Degli Abruzzi, 10129, Italy.
Fabrizio Dabbene and Chiara Ravazzi are with the Institute of Electronics, Computer and Telecommunication Engineering, National Research Council of Italy (CNR-IEIIT), c/o Politecnico di Torino, Corso Duca Degli Abruzzi, 10129, Italy. 
E-mail: chiara.ravazzi@ieiit.cnr.it, fabrizio.dabbene@ieiit.cnr.it.
Paolo Frasca is with Univ.\ Grenoble Alpes, CNRS, Inria, GIPSA-lab, F-38000 Grenoble, France and Research Associate at the IEIIT-CNR, Torino, Italy. E-mail:paolo.frasca@gipsa-lab.fr. \newline
\indent
The authors' research has been partially supported by the PRIN 2017 Project, Prot. 2017S559BB and by the French National Science Foundation ANR via project HANDY, number ANR-18-CE40-0010.}%
}

\newtheorem{corollary}{Corollary}
\newtheorem{theorem}{Theorem}

\newtheorem{proposition}{Proposition}

\newtheorem{problem}{Problem}

\theoremstyle{definition}

\newtheorem{example2}{Experiment}

\definecolor{darkcyan}{rgb}{0.0, 0.55, 0.55}
\definecolor{internationalorange}{rgb}{1.0, 0.31, 0.0}

\usepackage{color}

\usepackage{comment}
\usetikzlibrary{shapes, graphs,graphs.standard,quotes, calc, arrows.meta, fit, positioning} 
\tikzset{  
	-Latex,auto,node distance =1.5 cm and 1.3 cm, thick,
	state/.style ={ellipse, draw, minimum width = 0.9 cm}, 
	point/.style = {circle, draw, inner sep=0.18cm, fill, node contents={}},  
	bidirected/.style={Latex-Latex,dashed}, 
	el/.style = {inner sep=2.5pt, align=right, sloped}  
}  

\usetikzlibrary{decorations.pathmorphing}
\usepackage{amssymb}
\usepackage{nomencl}
\makenomenclature
\RequirePackage{ifthen}
\renewcommand{\nomgroup}[1]{%
	\ifthenelse{\equal{#1}{S}}{\item[\textbf{Symbols}]}{
		\ifthenelse{\equal{#1}{A}}{\item[\textbf{Acronyms}]}{}}}
\usepackage{hyperref}

\usepackage{tikz}
\usepackage{subfigure}

\begin{document}
\maketitle
\thispagestyle{empty}
\pagestyle{empty}
\begin{abstract}
We consider a dynamic model for competition in a social network, where two strategic agents have fixed beliefs and the non-strategic/regular agents adjust their states according to a distributed consensus protocol. We suppose that one strategic agent must identify $k_+$ target agents in the network in order to maximally spread its own opinion and alter the average opinion that eventually emerges. In the literature, this problem is cast as the maximization of a set function and, leveraging on the submodular property, is solved in a greedy manner by solving $k_+$ separate single targeting problems.
Our main contribution is to exploit the underlying graph structure to build more refined heuristics.
As a first instance, we provide the analytical solution for the optimal targeting problem over complete graphs. This result provides a rule to understand whether it is convenient or not to block the opponent's influence by targeting the same nodes. The argument is then extended to generic graphs leading to more accurate solutions compared to
a simple greedy approach. 
As a second instance,  by electrical analogy we provide the analytical solution of the single targeting problem for the line graph and derive some useful properties of the objective function for trees. Inspired by these findings, we define a new algorithm which selects the optimal solution on trees in a much faster way with respect to a brute-force approach and works well also over tree-like/sparse graphs.
The proposed heuristics are then compared to zero-cost heuristics on different random generated graphs and real social networks. 
Summarizing, our results suggest a scheme that tells which algorithm is more suitable in terms of accuracy and computational complexity, based on the density of the graphs and its degree distribution.
\end{abstract}

\section{Introduction}

In the course of the last decade, numerous works have considered the problem of optimally allocating resources to influence the outcome of opinion dynamics. The problem has attracted researchers with backgrounds from economics to engineering, who have deployed tools from game theory, optimization and, of course, network science~\cite{Kempe2003maximizing,hung2011optimization,Yildiz:2013,8395011,grabisch2018strategic,yi2020shifting}.
A large part of this research assumes a linear model of opinion evolution, as per the influential De Groot model of opinion evolution--see \cite{AP-RT:17,AP-RT:18} for a contextualization of this model. Under De Groot model, the steady-state opinions satisfy a linear equation defined by a weighted Laplacian matrix associated to the network graph.

A typical setup considers two strategic agents, holding extreme opinions (say, $-1$ and $+1$), which compete with the purpose of swaying the average steady-state opinion towards their own. This setup has the mathematical advantage of yielding an objective function that is a linear function of the node opinions. 

Actually, two kinds of (closely related) problems have been considered: in one formulation of the problem ({\it internal} influence), strategic agents have the opportunity to ``recruit'', among the regular nodes, influencers that hold their fixed opinions~\cite{hung2011optimization}. In another formulation ({\it external} influence), the strategic agents have the opportunity to create additional edges between themselves and target nodes~\cite{grabisch2018strategic}. Both setups allow for either game theoretic analysis, where the focus is on the interplay between the two strategic players, as well as optimization approaches where one of the players has a fixed strategy and the other is optimizing her strategy by targeting $k$ nodes for internal or external influence.

As both internal and external influence problems are combinatorially hard~\cite{yi2020shifting}, effective heuristics are needed. Most methods rely on submodularity to advocate 
greedy heuristics that reduce the general problem of targeting the $k_+$ best nodes to a sequence of $k_+$ problems of targeting the best node. Such an approach requires $O(Nk_+)$ evaluations of the equilibrium opinions (where $N$ is the number of network nodes). Each 1-best problem can be easily solved by comparing the $N$ possible solutions: in turn, evaluating each solution requires the solution of a linear system of equations, which can be performed in $O(M)$ operations (where $M$ is the number of network edges~\cite{nkV:2012}). Therefore, this kind of greedy approach typically results in $O(MN)$ cost, with the guarantee of a bounded error. Other heuristic approaches may achieve $O(M)$ cost, though without bounded-error guarantees~\cite{Vassio,rossi2018convergence}.

In this paper, we concentrate on the problem in which one of the two strategic agents has to optimize the deployment of $k_+$ additional links between herself and $k_+$ regular nodes (to which she is not yet connected). On this well studied problem, we provide theoretical results on specific networks.
First, we derive a closed-form solution for the Optimal Targeting Problem (OTP) over complete graphs leading to a zero-cost rule for optimal strategy.
Then, by electrical analogy, we provide the analytical solution for the Single Targeting Problem (STP) over line graphs and some of the properties of the objective function are extended to the branches of generic tree graphs. These theoretical findings allow us to design new 
algorithms for general graphs.
These heuristics are compared with optimal solution and zero-cost strategies, consisting in targeting with highest degree nodes. 
To put the results into perspective, we provide a scheme suggesting which is the best heuristic based on the cost vs accuracy trade-off, and the underlying graph.


\paragraph{Paper outline}
In Section \ref{sec:problem_formulation} the model of competition and the OTP are formally introduced. In Section \ref{OTP:Blocking} we derive an explicit solution of the OTP on the complete graph and propose a simple heuristic that requires no evaluations of the equilibrium opinions.
Section \ref{sec:STP_sp_net} presents some analytical results for STP on the line graph and on trees. These results lead to a heuristic criterion to accelerate the solution to the 1-best problem by avoiding the evaluation of all $N$ possible solutions (Section~\ref{sect-tree-like}). Finally, Section \ref{sec:conclusions} collects some concluding remarks. 

\paragraph{Notation}
Throughout this paper, we use the following notation. 
The set of real numbers is denoted by $\real$ and the set of non-negative integers is denoted by $\integernonnegative$. We denote column vectors with lower case letters and matrices with upper case letters.  The vector of all ones of appropriate dimension is represented by $\ind$.
We denote the 2-norm of a vector $x$ with the symbol $\|x\|$.  
Given a matrix $A$, $A^{\top}$ denotes its transpose. Moreover, $\mathrm{sr}(A)$ is the spectral radius of the matrix $A$, and a square matrix $A$ is said to be Schur stable if $\mathrm{sr}(A)<1$. A matrix $A$ with positive entries is said to be row stochastic if $A \ind = \ind$, and it is said to be row substochastic if $A \ind \leq  \ind$, where the inequality is entry-wise.

We represent the network by a directed graph, a pair $\G = (\V,\E)$, where $\V$ is the set of nodes, unitary elements of the network, and $\E\subseteq \V\times\V$ is the set of edges or links representing the relationships among such entities. A path in a graph is a sequence of edges which joins a sequence of vertices. A directed graph $\G$ is called strongly connected if there is a path from each vertex in the graph to every other vertex.  An undirected graph in which any two vertices are connected by exactly one path is called tree.
Given a matrix $W\in\R^{\V\times\V}$ with non-negative entries, the weighted graph associated to $W$ is the graph $\G = (\V , \E, W )$ with node set $\V$, defined by drawing an edge $(i,j)\in\E$ if and only if $W_{ij}>0$ and putting weights $W_{ij}$. 
If $ W $ is symmetric, i.e. $ W_{ij}  =W_{ji}$ for each $ i,j\in\mathcal{V} $, the \textit{undirected} edges will be denoted as unordered pairs $ \{i,j\} $, corresponding to both the directed links $ (i,j) $ and $ (j,i) $. A subset of nodes $ \mathcal{U}\subset\mathcal{V} $ is said to be \textit{globally reachable} in $ \mathcal{G} $ if for every node $ j\in\mathcal{V}\setminus\mathcal{U} $ there exists a path from $ j $ to some node $ i \in \mathcal{U}$.
Let $ \mathcal{G=(V,E},W) $ be a graph, then the \textit{in-neighborhood} of a node $ i \in \mathcal{V} $ is defined as $\mathcal{N}_i = \{j\in\mathcal{V} : (j,i) \in\mathcal{E} \}$. The \textit{in-degree} of a node is defined as $d_i= \sum_{j\in\mathcal{V}} W_{ij}$.
We will consider the \textit{normalized weight matrix} $Q=D^{-1}W$ where $D$ is the diagonal matrix with diagonal entries equal to the \textit{in-degree} of node $i\in\V$: $D_{ii}= \sum_{j\in\mathcal{V}} W_{ij}$. We will denote the \textit{Laplacian matrix} by $L=D - W$.
\section{OTP and its state-of-the-art solutions}\label{sec:problem_formulation}
\subsection{Dynamic model for competition}
We consider an influence network described by a graph $ \mathcal{G} =(\mathcal{V,E},W) $. Nodes $v\in \mathcal{V} $ represent the agents and $\mathcal{E} $ is the set of edges describing the interactions among them. The structure of the network is encoded in the adjacency matrix $W$. We assume that the graph is undirected and we consider the normalized weight matrix  $Q=D^{-1}W$. We assume that the set of nodes is partitioned into two disjoint sets: $\V=\cal{R}\cup\cal{S}$, where $\cal{R}$  and $\cal{S}$ are the set of regular and strategic agents, respectively. 

We assume that each agent is endowed with a state $x_s(t)=x_{s}\in\{-1,1\}$ for all $s\in\cal{S}$ and $x_i(t)\in [-1,+1] $ for all $i\in\cal{R}$, representing the opinion/belief at time $ t $. At each time step $ t\in\integernonnegative$ the opinion of a regular agent $i\in\mathcal{R}$ is updated as a response to the interaction with the neighbors, according to the following rule 
\begin{align}
x_i(t+1)=\sum_{j\in\mathcal{V}}Q_{ij}x_j(t) 
\label{eq:linearaveraging}
\end{align}
where $Q_{ij} \geq0 $ for all $i\in\cal{R}$ and for all $j \in\cal{V}$, $Q_{ij} =0\iff (i,j) \notin \cal{E}$ and $\sum_{j\in\V}Q_{ij}=1$ for all $i\in\cal{R}$.

Assembling opinions of regular and strategic agents in a vector $x(t)=({x^{\cal{R}}}^\top(t),{x^{\cal{S}}(t)}^{\top})^{\top}=({x^{\cal{R}}}^{\top}(t),{x^{\cal{S}}}^{\top})^{\top}$, we can rewrite the dynamics in the following compact form
\begin{align}
x(t+1)=Qx(t)\hspace{1cm}t=1, 2, \dots
\label{eq:linearaveraging2}
\end{align}
with $$Q=\left(\begin{array}{cc}Q^{11}&Q^{12}\\
0&I
\end{array}\right)=\left(\begin{array}{cc}(D^{11})^{-1}W^{11}&(D^{11})^{-1}W^{12}\\
0&(D^{22})^{-1}
\end{array}\right)
$$
where matrices $Q^{11}$, $Q^{12}$,  $W^{11}$, $W^{12}$,  $D^{11}$, and $D^{12}$, are nonnegative matrices of appropriate dimensions.
Such equation in the social science context is known as the \textit{DeGroot} opinion dynamics model \cite{Degroot:74,French_1956} or, more generally, as the \textit{linear averaging dynamics} on $ \mathcal{G}. $
We assume that each strategic agent has at least one link to a regular agent, but no more than one to the same target. Hence, in the paper, we make the standing assumption 
that the influence matrix $ W^{12}  \in \{0,1\}^{\mathcal{R\times S}} $ is such that $\mathbbm{1}^\top W^{12} \geq \mathbbm{1}^\top $, and that $\mathcal{G}|_\mathcal{R}$, i.e.\ the graph  restricted to nodes in $\mathcal{R}$,  is strongly connected.

The following proposition holds \cite{ME2010}.
\begin{proposition}\label{prop:convergence}
Let $Q^{11}$ be substochastic and asymptotically stable. Then $ I-Q^{11} $ will be invertible with non-negative inverse matrix.
Moreover, for every initial state vector $ x(0) \in\mathbb{R}^{\mathcal{V}}$, the dynamics in (\ref{eq:linearaveraging2}) converges to a finite limit profile $\overline{x}=((\overline{x}^{\cal{R}})^{\top},(\overline{x}^{\cal{S}})^{\top})^{\top}$
		\begin{equation}\label{eq:equilibria}\overline{x}^{\cal{R}}= \lim_{t\to +\infty} x^\mathcal{R}(t)=(I-Q^{11})^{-1}Q^{12}x^\mathcal{S}. \end{equation}
\end{proposition}
Proposition \ref{prop:convergence} states that the opinions converge to a stationary profile that is a combination of the opinions of all strategic agents.

\subsection{Optimal Targeting Problem}
\label{sec:OptimalTargeting}
In this paper we consider the situation where there are $N$ regular agents, indexed by $i\in\mathcal{R}=\{1,\ldots,N\}$ and two strategic
agents $\mathcal{S}=\{N+1,N+2\}:=\{{{\piu}},{{\meno}}\}$ with opinion $x_{\piu}(t)=+1$ and  the latter with opinion $x_{\meno}(t)=-1$ respectively, for all $t\in\integernonnegative$. 
We investigate how to identify regular nodes in $\cal{R}$ in
order to maximize the influence of opinion +1 on the
final limit profile, assuming that edges of the strategic agent $\meno$ are already placed. 
We use $(\bar{x}_v^{(\mathcal{A})})_{v\in\mathcal{R}}$ to denote the asymptotic opinion profile that emerges from the particular configuration in which the nodes belonging to the set $\mathcal{A}$ are additionally linked to strategic node $\piu$. The OTP is defined as the following optimization problem.
\begin{problem}[Optimal Targeting Problem ($k_+$OTP)]\label{prob:OTP}
Given $\mathcal{G=(\V,\E)}$, let $\mathcal{A}^-=\{v\in\mathcal{R}:(\meno,v)\in\mathcal{E}\}\neq\emptyset$, $\mathcal{A}^{(0)}=\{v\in\mathcal{R}:(\piu,v)\in\mathcal{E}\}$. Find the node set $\mathcal{A}^+$
\begin{align}
\mathcal{A}^+\subseteq\argmax{\mathcal{A}\subseteq\mathcal{R}\setminus\mathcal{A}^{(0)}:\ |\mathcal{A}|\leq k_+} F_+(\mathcal{A}),
\label{eq:optprob}
\end{align}
with
\begin{align}F_+(\mathcal{A}) = \frac{1}{N}\sum_{v\in\mathcal{R}}\bar{x}_v^{(\mathcal{A})} \hspace{0.3cm},\hspace{0.3cm} \mathcal{A}\subseteq\mathcal{R}\setminus\mathcal{A}^{(0)}.
\label{eq:objfun}
\end{align}
\end{problem}
In this optimization problem, for any different choice of the set $\mathcal{A}$,
the influence matrices $Q^{11}$ and $Q^{12}$ change and the final limit profile needs to be computed, requiring a new matrix inversion. Then, the complexity of the problem is combinatorial since we need to find the best solution among all ${N\choose k_+}$ possible configurations. 

\begin{problem}[Single Targeting Problem (STP)]
The specific case where $ k_+=1 $, $|A^-|=1$, will be referred to as \textit{Single Targeting Problem} (STP). Then, the OTP reduces to finding the node that maximizes the following objective function:
$\max_{v\in\mathcal{R}{\setminus\mathcal{A}^{(0)}}} F_+(\{v\}).$
\end{problem}

The OTP described in Problem \ref{prob:OTP} is computationally challenging if we are interested in targeting simultaneously $k_+\geq 2$ nodes in the network. An heuristic based on the out-degree centrality, defined as the number of outgoing links of the nodes, is a rough but common approach to approximate the OTP solution. This method, summarized in Algorithm \ref{alg:zerocost_alg}, consists in selecting the $ k_+ $ nodes with highest degree (if there are more subsets with this property, one of them is selected randomly). 
 It should be noticed that this is a zero-cost heuristics, in the sense that it provides a strategy without the burden of the equilibrium opinions' computation. This simple heuristic will be used as a benchmark for the proposed methods.

\begin{algorithm}   [h]\caption{Degree Heuristic for $k_+$OTP}
	\begin{algorithmic} \label{alg:zerocost_alg}
		\REQUIRE $\mathcal{G}=(\mathcal{V},\mathcal{E})$ graph, number of available links $k_+$\\
		
		{\bf Initialization:}
		\STATE  $\quad$ $ \mathcal{D} $ set of nodes with top-$ k_+ $ degree
		\STATE  $\quad$ $\mathcal{A}_{k_+} = \mathcal{D}$
		\RETURN $\mathcal{A}_{k_+}$, $F_+(\mathcal{A}_{k_+})$
	\end{algorithmic}
\end{algorithm}

Another common approach in literature is to solve the optimization problem for one target at a time in a greedy manner, i.e. choosing at each iteration a target which gives the largest gain in the objective function. This approach allows to reduce the complexity and can be applied in large social networks. We review the procedure in Algorithm \ref{alg:Greddy_alg}.

\begin{algorithm}\caption{Greedy algorithm for $k_+$OTP}
\begin{algorithmic} \label{alg:Greddy_alg}
\REQUIRE $\mathcal{G}=(\mathcal{V},\mathcal{E})$ graph, number of available links $k_+$\\

{\bf Initialization:}
\STATE  $\quad$ $ \mathcal{A}_0=\emptyset $ 
\FOR{$ i\in\{1,\ldots,k_+\} $}
\STATE $\mathcal{A}_i = \mathcal{A}_{i-1} \cup \argmax{v}\{\Delta(v | \mathcal{A}_{i-1}). \}$
\ENDFOR
\RETURN $\mathcal{A}_{k_+}$, $F_+(\mathcal{A}_{k_+})$
\end{algorithmic}
\end{algorithm}

The greedy algorithm starts with the empty set $\mathcal{A}_0=\emptyset$, and at iteration $i$ it adds a new element maximizing the discrete derivative $\Delta(v | \mathcal{A}_{i-1})$
$$\mathcal{A}_i = \mathcal{A}_{i-1} \cup \argmax{v}\{\Delta(v | \mathcal{A}_{i-1}). \} $$
where
$\Delta(v | \mathcal{A})=F_+(\mathcal{A}\cup\{v\})-F_+(\mathcal{A})$.
\begin{theorem}\label{thm:submodular_function_properties}
For any arbitrary instance of $\G=(\V,\E),$ the set function $F_+(\mathcal{A})$ defined in \eqref{eq:objfun} is monotone and submodular. 
\end{theorem}
A proof of this fact can be found in the recent report \cite{yi2020shifting}, or it can be obtained from results in \cite{DBLP:journals/corr/abs-1806-11253}, \cite{8395011}, and \cite{yi2020shifting}, see \cite{tesi-bini} for details.
Submodularity ensures that the greedy algorithm provides a good approximation to the optimal solution \cite{krause14survey}.
\begin{corollary}For any positive $k_+$ and $\ell$
$$F_+(\mathcal{A}_{\ell}) \geq \left(1 - 1/\mathrm{e}^{-\ell/k_+}\right) F^{\star}$$
with $F^{\star}=\max_{|\mathcal{A}|\leq k_+} F_+(\mathcal{A}).$
\end{corollary}
Algorithm \ref{alg:Greddy_alg} instead of evaluating the objective function for all possible combination of edges, chooses one edge at a time, reducing significantly the complexity from $O({M}(N/k_+)^{k_+})$ to $O({M}Nk_+)$, at the price of a bounded relative error $|F^{\star}-F_+(\mathcal{A}_{\ell})|/|F^{\star}|\leq1/\mathrm{e}$.

\section{OTP: a blocking approach}\label{OTP:Blocking}
In this section, we study the OTP in the complete graph. Inspired by this result, we propose a simple heuristic to solve OTP in general graphs that does not require any evaluation of the equilibrium opinions.

\subsection{OTP in Complete Graphs}
\label{sec:CompleteGraph}
\begin{figure}
\label{fig:CompleteGraph}
	\centering
		\centering\resizebox{0.69\columnwidth}{0.37\columnwidth}{%
	\begin{tikzpicture}
	[-,shorten >=1pt,auto,node distance=2cm,main node/.style={circle,fill=blue!20},square/.style={regular polygon,regular polygon sides=4,inner sep=0.15em}]
	
	\tikzset{edge/.style = {->,> = latex'}}
	
	\graph[nodes={circle}, clique, n=10, clockwise, radius=3cm]
	{
		1/"$v^{\emptyset}$"/[fill=blue!20], 2/"$v^+$"/[fill=green!20], 3/"$v^+$"/[fill=green!20], 4/"$v^{\emptyset}$"/[fill=blue!20], 5/"$v^{\pm}$"/[fill=yellow!20],6/"$v^{\pm}$"/[fill=yellow!20],7/"$v^{\emptyset}$"/[fill=blue!20],8/"$v^{\emptyset}$"/[fill=blue!20],9/"$v^{\emptyset}$"/[fill=blue!20],10/"$v^-$"/[fill=red!20]
	};

	\node[main node,square,fill=red!40] (-) [midway, left = 6cm] {$-$};
	\node[main node,square,fill=green!40] (+) [midway, right = 6cm] {$+$};
	
	\draw[edge,red] (10) -- (-);
	
	\draw[edge,green] (3) -- (+);
	\draw[edge,green] (2) -- (+);
	
	\draw[edge,green] (5) -- (+);
	\draw[edge,red] (5) -- (-);
	\draw[edge,green] (6) -- (+);
	\draw[edge,red] (6) -- (-);
	
	\end{tikzpicture}}
	\caption{The regular nodes form a complete graph with $ N=10 $, $p=2$, $ q=1$, and $ r=2 $.}
	\label{fig:complete}
\end{figure}
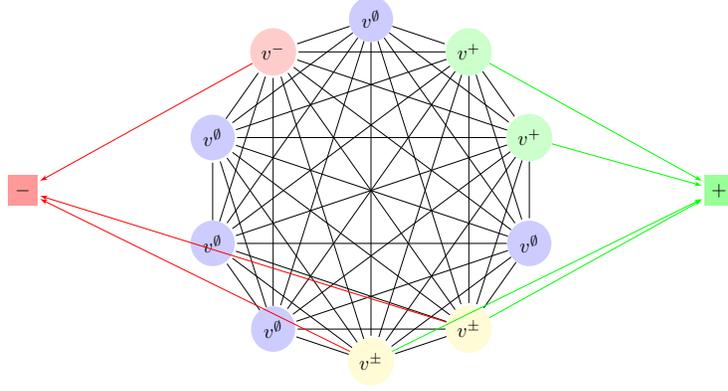
We consider the situation where there are $N$ regular nodes in the network forming a complete graph.
In order to compute the objective function in this case, we exploit the anonymity property, i.e. the fact that regular agents share the same neighborhood, except for being connected or not to the strategic agents. Based on the latter, we can distinguish among four kinds of regular nodes (see Figure~\ref{fig:complete}), that is, we partition the set $\mathcal{R}$ into
\begin{itemize}
	\item $ \mathcal{R}^{+}$, the set of nodes linked to $\tiny \piu $ but not to $\tiny \meno $ and denote $p:=|\mathcal{R}^{+}|$ (see green nodes in Figure \ref{fig:complete});
	\item $ \mathcal{R}^- $, the set of nodes linked to $\tiny \meno $ but not to $\tiny \piu $ and denote $q:=|\mathcal{R}^{-}|$  (see pink nodes in Figure \ref{fig:complete});
	\item $ \mathcal{R}^{\pm}$, the set of nodes linked to both $\tiny \piu $ and $\tiny \meno $ and denote $r:=|\mathcal{R}^{\pm}|$  (see yellow nodes in Figure \ref{fig:complete}); and 
	\item $ \mathcal{R}^{\emptyset}$, the set of nodes linked to neither of them, so that $N-p-q-r=|\mathcal{R}^{+}|$  (blue nodes in Figure \ref{fig:complete}).
\end{itemize}
Anonymity ensures that the objective function is only a function of $p,q,r$, that is, we can write $F_+(\mathcal{A})=:\mathsf{F}_+(p,q,r)$.
Then, the system of equations in \eqref{eq:equilibria} becomes 
\begin{align*}
\begin{cases}
\bar{x}_{v^+} = \frac{p-1}{N}\bar{x}_{v^+}+\frac{q}{N}\bar{x}_{v^-}+\frac{r}{N}\bar{x}_{v^{\pm}}+\frac{N-p-q-r}{N}\bar{x}_{v^{\emptyset}}+\frac{1}{N}\\
\bar{x}_{v^-} = \frac{p}{N}\bar{x}_{v^+}+\frac{q-1}{N}\bar{x}_{v^-}+\frac{r}{N}\bar{x}_{v^{\pm}}+\frac{N-p-q-r}{N}\bar{x}_{v^{\emptyset}}-\frac{1}{N}\\
\bar{x}_{v^{\pm}} = \frac{p}{N+1}\bar{x}_{v^+}+\frac{q}{N+1}\bar{x}_{v^-}+\frac{r-1}{N+1}\bar{x}_{v^{\pm}}+\frac{N-p-q-r}{N+1}\bar{x}_{v^{\emptyset}}\\
\bar{x}_{v^{\emptyset}} = \frac{p}{N-1}\bar{x}_{v^+}+\frac{q}{N-1}\bar{x}_{v^-}+\frac{r}{N-1}\bar{x}_{v^{\pm}}+\frac{N-p-q-r-1}{N-1}\bar{x}_{v^{\emptyset}}
\end{cases}
\end{align*}
with $v^{+}\in\mathcal{R}^+$, $v^{-}\in\mathcal{R}^-$, $v^{\pm}\in\mathcal{R}^{\pm}$, and $v^{\emptyset}\in\mathcal{R}^{\emptyset}$.
Solving this system, we find 
\begin{align*}
\mathsf{F}_+(p,q,r) 
&= \frac{(N+2)(p-q)}{N(N+2)(p+q)+2N(N+1)r}
\end{align*}
Notice that if $p=q$, then $\mathsf{F}_+(p,q\!\!\!=\!\!\!p,r)= 0$ and that $\mathsf{F}_+(p,q,r)$ is decreasing in $r$. 
\smallskip
This formula allows us to give an explicit solution to the OTP problem. Let $p_0,q_0,r_0$ be the number of regular nodes initially linked to strategic agent $\tiny \piu $ but not to $\tiny \meno $, to $\tiny \meno $ but not to $\tiny \piu $, and to both, respectively. Strategic agent $\tiny \piu $ has $k_+$ available links to add and define
\begin{align*}
\mathsf{F}^{\star}_+:=&\max_{\mathcal{A}\subseteq\mathcal{R}\setminus \mathcal{A}^{(0)}: |\mathcal{A}|\leq k_+}{F}_+(\mathcal{A}).\end{align*}
Let $r_1$ and $p_1$ be the numbers of additional nodes that are targeted by strategic agent $\tiny \piu $ and initially are, respectively, linked or not to strategic agent $\tiny \meno $ (with the constraints that $p_1+r_1=k_+$, $r_1\le q_0$ and $p_1\le N-p_0-q_0-r_0$). Observe that
$F_+(\mathcal{A})=\mathsf{F}_+(p_0 +p_1,q_0 - k_+ +p_1,r_0+k_+-p_1).$

\begin{proposition}[$k_+$OTP on complete graph]\label{OTP:complete}
The optimal solution $p_1^\star$ satisfies the following properties.
\begin{itemize} 
\item If $k_+<q_0-p_0$. then $p_1^\star=k_+$ and
$\mathsf{F}^{\star}_+=\mathsf{F}_+(p_0+k_+,q_0,r_0)$;
\item If $ k_+=q_0-p_0$, then $\mathsf{F}_+(p,q,r)=0$ irrespective of $p_1$;
\item If $  k_+>q_0-p_0$ then $p_1^{\star}= \max\{0, k_+ - q_0\}$ and 
\begin{align*}\mathsf{F}^{\star}_+&=\mathsf{F}_+(p_0 +p_1^{\star},q_0 - k_++p_1^{\star},r_0+k-p_1^{\star})\end{align*}
\end{itemize}
\end{proposition}

\begin{proof}
If $k_+$ edges are available, then the strategic agent $\tiny \piu $ can target $p_1$ nodes not already linked to $\tiny \meno $ and $r_1$ nodes in $\mathcal{R}^{-}$.
Adding these links, we obtain that
\begin{align*}
\mathsf{F}_+^{\star}
&=\max_{p_1,r_1:p_1+r_1= k_+}\mathsf{F}_+(p_0+p_1,q_0-r_1,r_0+r_1).
\end{align*} 
The objective function $\mathsf{F}_+ $ can be increasing or decreasing in $ p_1 $, depending on whether $p _0-q_0 \le k_+ $.
If $k_+> q_0-p_0$, then the objective function is negative and the maximizing value $p_1^{\star}$ is obtained with $p_1^{\star}=k_+$ and $r_1^{\star}=0.$ 
If $k_+= q_0-p_0$, then $\mathsf{F}_+ $ is always zero.
If $k_+< q_0-p_0$, then the objective function is positive and the optimum is reached by taking the smallest value of $p_1$, that is, the largest value of $r_1.$ Since the latter is naturally constrained by  $r_1\le q_0$ and by $ p_1\le N - p_0 - q_0 - r_0$, the result follows.
\end{proof}

Proposition \ref{OTP:complete} asserts that if the sum of available links and nodes already connected to strategic agent $\tiny\piu$, not targeted by strategic node $\tiny\meno$, does not exceed the number of nodes connected to $\tiny \meno $ but not to $\tiny \piu $, then the optimal strategy is to use all the available budget to target the nodes not already linked to $\tiny \meno $.
Otherwise, the optimal strategy is to use a portion of the budget to target all the nodes linked to $\tiny \meno $, and the extra budget to influence the remaining nodes.
%
%

\subsection{Blocking heuristics on general graphs}
\label{sec:BlockingHeuristics}
In Section \ref{sec:CompleteGraph} the OTP has been solved for the complete graph. It is worth remarking that, using the greedy method (see Algorithm \ref{alg:Greddy_alg}), at each iteration the algorithm would have introduced an error if the condition $ p^{(k)} + 1 > q^{(k)}$ were not satisfied, where $ p^{(k)} $, $ q^{(k)} $ are the number of nodes connected exclusively to $\piu$ and $\meno$ at iteration $ k =1,\dots, k_+$, respectively.
With this in mind, we design a new heuristic in  Algorithm \ref{alg:block_alg1}. Let us denote by $ \mathcal{A}^- $ the set of nodes linked to $ \meno $, while by $ \mathcal{A}^{(0)} $ the set of initial nodes linked to $ \piu $.

\begin{algorithm}   [h]\caption{Blocking Heuristic for OTP}
	\begin{algorithmic} \label{alg:block_alg1}
		\REQUIRE $\mathcal{G}=(\mathcal{V},\mathcal{E})$ graph, set node $\mathcal{A}^{-}$, set node $\mathcal{A}^{(0)}$, number of available links $k_+$\\
		
		{\bf Initialization:}
		\STATE  $\quad$ $ \mathcal{A}_0=\emptyset $ , $ s=0 $
		\IF{$ k_+ > |\mathcal{A}^{-}\setminus\mathcal{A}^{(0)}| - |\mathcal{A}^{(0)}\setminus\mathcal{A}^{-}|$}
		\STATE $ \mathcal{A}_s=\mathcal{A}^{-}\setminus\mathcal{A}^{(0)} $ , $ s = |\mathcal{A}^{-}\setminus\mathcal{A}^{(0)}| $
		\ENDIF
		\IF{$ k_+>s $}
		\FOR{$ i\in\{s+1,\ldots,k_+\} $}
		\STATE $\mathcal{A}_i = \mathcal{A}_{i-1} \cup \argmax{v\in\mathcal{R}}\{\Delta(v | \mathcal{A}_{i-1}) \}$
		\ENDFOR
		\ELSE 
		\STATE $\mathcal{A}_{k_+} = \mathcal{A}_{s} $
		\ENDIF
		\RETURN $\mathcal{A}_{k_+}$, $F_+(\mathcal{A}_{k_+})$
	\end{algorithmic}
\end{algorithm}

Algorithm \ref{alg:block_alg1}, in practice, compares the $ |\mathcal{A}^{(0)}\setminus\mathcal{A}_{-}| +k_+ $ overall edges of agent $ \piu$ not linked to $ \meno $, with the $ |\mathcal{A}^{-}\setminus\mathcal{A}^{(0)}| $ edges not linked to $\piu$ of agent $\meno $. If the comparison tells that the former is larger, then extra edges of $\piu$ will be used until possible to target the nodes in $ \mathcal{A}^{-}\setminus\mathcal{A}^{(0)} $. Then, if some edges are still available to agent $ \piu $, they will be placed by following the greedy approach (see Algorithm \ref{alg:Greddy_alg}). On the other hand, if the first condition is not satisfied, it simply reduces to the Greedy Heuristic.
It should be noticed that Algorithm \ref{alg:block_alg1} has in general a smaller cost than the Greedy Heuristic, since it requires to evaluate the equilibrium opinions $\max\{0,k_+-q\} N$ times.

\begin{example2}
\label{ex:ERheuristics}
We compare the proposed algorithms with zero-cost heuristics and greedy approaches. The study is aimed at highlighting the effect of the structure of the network on the performance of the algorithms. In particular, we consider random generated Erdos-Renyi graphs  \cite{Bollobas1998Modern} with parameters $ N = 400$ and $ p=a\frac{\log N}{N} $. For each value of $a\in[1.5,10]$, we generate $ 50 $ random graphs and we connect $ |\mathcal{A}^-|= 3 $ nodes randomly to the strategic agent~$ \meno $, whereas strategic agent $\piu$ has $ k_+=5 $ available budget. 
It should be noted that the best performance is obtained with the heuristics based on the Blocking approach (Alg.~3) and, as to be expected, the Degree heuristics (Alg.~1) performs the worst.
	\begin{figure}
		\centering
		\includegraphics[trim=0cm 0cm 0cm 0cm, clip,width=0.78\columnwidth]{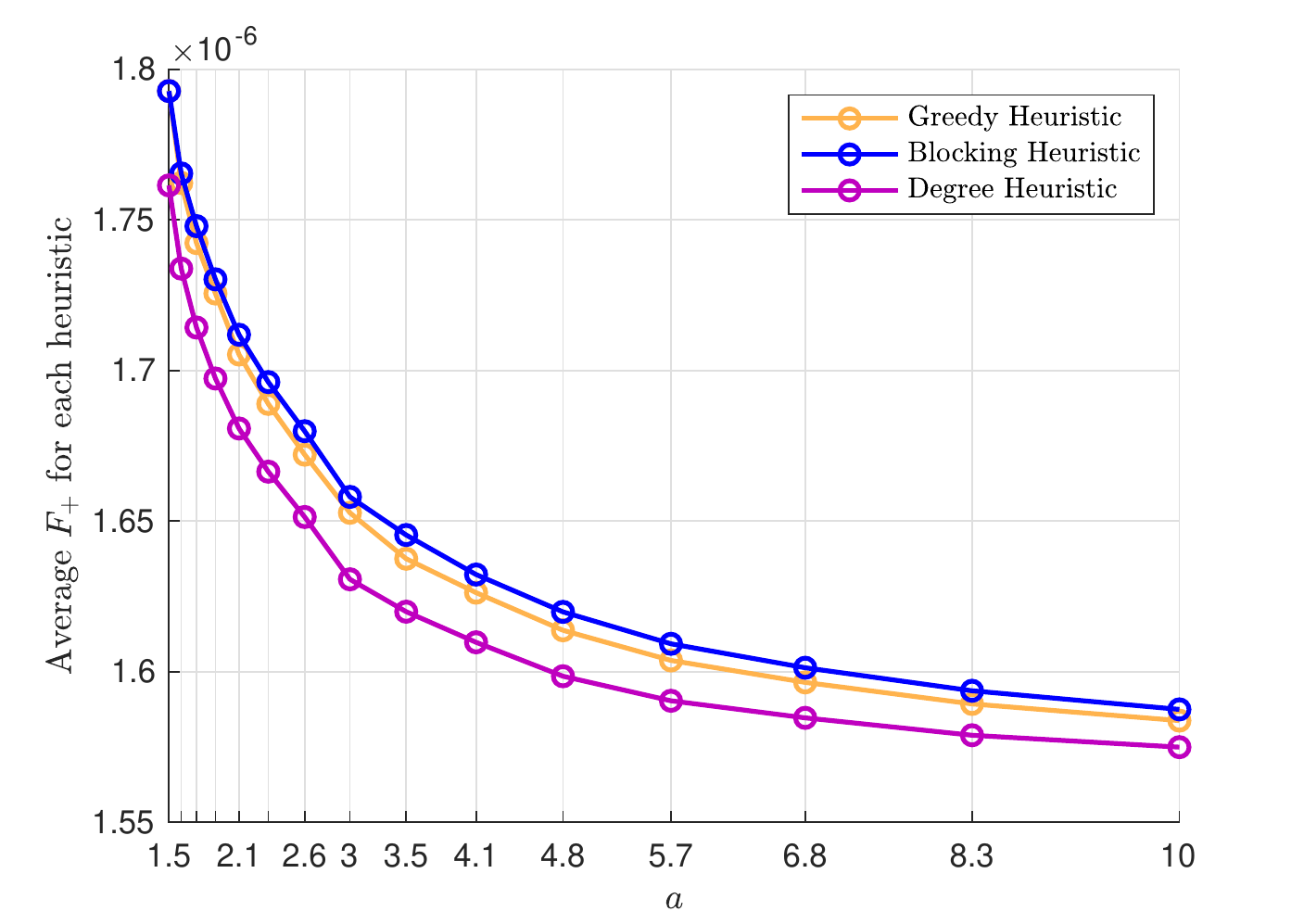}
		\caption{Erdos-Renyi with $p = a\log(N)/N$: Average $F_+$ over 50 simulations as a function of connectivity parameter $a$.}
		\label{fig:dense4}
	\end{figure}
	\end{example2}
	
	
\section{STP: Electrical analogy and trees}\label{sec:STP_sp_net}
In this section, we show some analytical results regarding the solution of STP on specific networks:  Line Graphs and Trees.
Before stating and proving our results, we describe a key methodology.

\subsection{Electrical Network Analogy}
In order to solve analytically the STP, it is convenient to use the electrical network analogy as presented in \cite{Vassio}. 
Let us briefly recall the basic notions of such analogy. 

We consider a strongly connected undirected graph $ \mathcal{G}=(\mathcal{V},\mathcal{E},W) $, where $ \mathcal{E} $ is the set of unordered couples $ \{i,j\} $. Such graph can be seen as an electrical network $ \mathcal{G_C}=(\mathcal{V},\mathcal{E},C) $ where the weight matrix $ W $ is replaced by the conductance matrix $ C \in\mathbb{R}^{\mathcal{V}\times\mathcal{V}}$, where $ C_{ij} =C_{ji} $ is now the conductance between the nodes $ i $ and $ j $ (notice how the reciprocity assumption must hold). Then, let us define
 the \textit{incidence matrix}  $ B\in\{0,+1,-1\}^{\mathcal{E}\times\mathcal{V}} $, such that $	B\mathbbm{1}=0$ and $B_{ei}\neq 0\iff i \in e\text{ with }e\in\mathcal{E}$. It is straightforward to verify that given $e=\{i,j\} $ , the $e$-th row of $B$ has all entries equal to zero except for $B_{ei}$ and $B_{ej}$: one of them will be $+1$ and the other one $-1$.
Let $D_C\in \mathbb{R}^{\mathcal{E}\times\mathcal{E}}$ be the diagonal matrix whose entries are $ (D_C)_{ee}=C_{ij}=C_{ji}\text{ with }e=(i,j)\in\mathcal{E}$. It should be noted that
	$
	B^\top D_C B=D_{C\mathbbm{1}}-C\hspace{0.3cm}
	$
	where $ D_{C\mathbbm{1}} = \mathrm{diag}(C\mathbbm{1})$. Indeed $ D_C B $ associates at each row of B the weight of the corresponding edge multiplied by $ 1 $ or $ -1 $, while $ B^\top D_C B $ generates the matrix that on each  diagonal entry has the sum of all the conductances on such node, while on the $ ij $-th entry it has the conductance value of edge $ \{i,j\}$ of negative sign, if present.
	
Defining	
	 $ \eta\in\mathbb{R}^{\mathcal{V}} $ as the \textit{input current vector} (positive if ingoing, negative if outgoing),
	such that $ \eta^\top\mathbbm{1}=0 $; $ V\in\mathbb{R}^{\mathcal{V}} $ as the the \textit{voltage vector}, and $ \Phi\in\mathbb{R}^{\mathcal{E}} $ as the \textit{current flow vector} (positive if going from $ i $ to $ j $ on $ (i,j) $), then the usual Kirchoff and Ohm's law can be written as follows
\[
\begin{cases}
B^\top \Phi=\eta\\
D_C B V = \Phi
\end{cases}
\]
leading to
\begin{align}
L(C)V&=\eta
\label{eq:Lapl2}
\end{align}
where 
$ L(C) :=D_{C\mathbbm{1}}-C$ is the \textit{Laplacian} of $ C $. Since the graph is strongly connected, $L(C)$ has rank $|\mathcal{V}|-1 $ and $ L(C)\mathbbm{1}=0 $, making $ V $, up to translations, the unique solution of the system. Also notice that $ (L(C)V)_i=0\hspace{0.3cm}\forall i\in\mathcal{V}$ such that $ \eta_i=0 $. The Equation (\ref{eq:Lapl2}) resembles the system in Proposition \ref{prop:convergence}, where the asymptotic opinion of regular agents can be interpreted as voltages with $ 0 $ input current, while those the strategic nodes as voltages fixed to $1$ and $-1$ with input current different from $ 0 $.

From now on, we will exploit the electrical analogy where the agents are nodes in the electrical network and their asymptotic opinions are the associated voltages. In this analogy, the strategic nodes $\tiny{{\meno}} $ and $ \tiny{{\piu}}  $ 
are considered voltage sources of value $ -1 $ and $ +1 $ respectively.
Thus, the objective function of OTP becomes $$F_+(\mathcal{A})=\frac{1}{N}\sum_{i\in \mathcal{R}}V^{(\mathcal{A})}(i)$$
where $ V^{(\mathcal{A})}(i) $ is the voltage of node $ i $ when the set of nodes linked to $ \tiny{{\piu}}  $ is $ \mathcal{A} $.

In the sequel, we will also make use of two common operations that allow to replace an electrical network by a simpler one without changing certain quantities of interest. Since current
never flows between vertices with the same voltage, we can merge vertices having
the same voltage into a single one, while keeping all existing
edges, voltages and currents are unchanged ({\em{gluing}}, \cite{Bollobas1998Modern})
Another
useful operation is replacing a portion of the electrical network
connecting two nodes $h,k$ by an equivalent resistance, a
single resistance denoted as $R_{\text{eff}}$ which keeps the difference of voltages $V_h-V_k$ unchanged ({\em series and parallel laws} \cite{Bollobas1998Modern}).


\subsection{Line Graph}
\label{sec:LineGraph}

We denote by $ v^+ $ and $ v^- $ the regular nodes that are linked to the strategic nodes $  \piu$ and $ \meno $, respectively.

\begin{proposition}[STP on the Line]\label{prop:LineGraph}
	Assume the strategic node $ \meno $ is directly connected to a generic node {\color{black}$v^{-}=\ell$}. Then, the objective function reads
	$$F_{+}(k)=\begin{cases}\frac{-k^2+(N+1)k-(N+1)\ell+\ell^2}{N(k-\ell+2)}
	&\text{if }k\geq \ell
	\\
	\frac{-k^2+(N+1)k-(N+1)\ell+\ell^2}{N(\ell-k+2)}
	&\text{if }k< \ell
	\end{cases}
	$$
and attains the maximum value at $k^*= \argmax {}\{F_+(\lfloor \widehat{k} \rfloor), F_+(\lceil \widehat{k} \rceil)\}$ with {\color{black}
		\begin{align*}
		\widehat{k}=
		\begin{cases}
		\ell - 2+\sqrt{2N+6-4\ell} \hspace{0.5cm}&\text{if }\ell<\frac{N+1}{2}\\
		\ell + 2+\sqrt{4\ell +2 -2N}  &\text{if }\ell\ge\frac{N+1}{2}
		\end{cases}
		\end{align*}
	}
\end{proposition}

\begin{proof}
	Assume the strategic node $\meno$ is directly connected to a node $\ell$  and the strategic node $\piu$ will choose a node $k\geq \ell$. 
	By the electrical analogy, we can interpret the strategic nodes $\meno$ and $ \piu$ as voltage sources of value $ -1 $ and $ +1 $ respectively. The nodes $v\in\{1,\ldots,\ell-1\}$ and  $v\in\{k+1,\ldots,N\}$ will be short-circuited with the nodes $ \ell $ and $ k $, respectively, i.e. 
	\begin{gather*}
	V(-)=-1, \ V(+)=+1\\
	V(1)=V(2)=\dots=V(\ell)\\
	V(k)=V(k+1)=\dots=V(N)
	\end{gather*}
	\begin{figure}  
		\centering	    \resizebox{.60\columnwidth}{!}{
			\begin{tikzpicture}
			[-,shorten >=1pt,auto,node distance=2cm,thick,main node/.style={circle,fill=blue!20},square/.style={regular polygon,regular polygon sides=4,inner sep=0.15em},decoration={straight zigzag,meta-segment length=0.5cm,amplitude=0.3cm}]
			\tikzset{edge/.style = {->,> = latex'}}
			
			\node[main node,square] (-) {$-$};
			\node[main node] (l) [right of=-]{$\ell$};
			\node (d) [right of=l] {$ \dots $};
			\node[main node] (k) [right of=d] {$k$};
			\node[main node,square] (+) [right of=k] {$+$};
			\node (d1) [below of=l] {$ \vdots $};
			\node[main node] (1) [below of=d1] {1};
			\node (d2) [below of=k] {$ \vdots $};
			\node[main node] (N) [below of=d2] {$N$};
			\draw[decorate] (l) -- (d);
			\draw[decorate] (l) -- (d1);
			\draw[decorate] (d1) -- (1);
			\draw[decorate] (d) -- (k);
			\draw[decorate] (k) -- (d2);
			\draw[decorate] (d2) -- (N);
			\draw[decorate] (-) -- (l);
			\draw[decorate] (+) -- (k);

			\end{tikzpicture}}
		\caption{Circuit analogous Line Graph with $ v^-=\ell$, $ v^+=k$, $ k\ge l $}
		\label{fig:linegraph_vl_circuit}
	\end{figure}
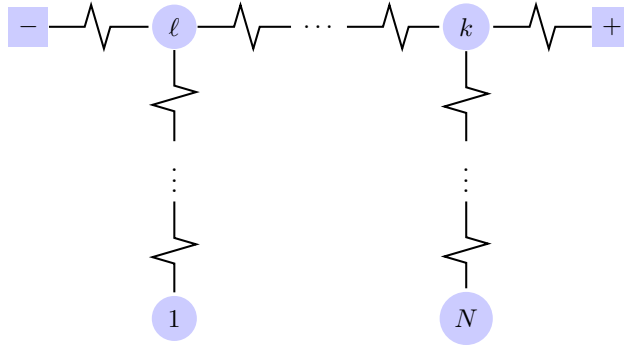
	\noindent
	We compute the voltage in each node $ i = \ell,\ell+1\dots,k$ as the voltage drop in the voltage divider (as represented below) where the effective resistances are the summation of the resistances on the left and on the right of node $ i $, that is
%
%
%
 $R^\text{eff}_{-i}=i-\ell+1$
	and $R^{\text{eff}}_{i+}=k-i+1$
	leading to
	\begin{align*}
	V(i)-V(-)&=(V(+)-V(-))\frac{i-\ell+1}{i-\ell+1+k-i+1}\\
	V(i)&=2\frac{i-\ell+1}{k-\ell+2}-1
	\end{align*}
	and
	\begin{align*}
	F_+(k)&=\frac{1}{N}\Bigl[\ell V(\ell)+\sum_{i=\ell+1}^{k-1}V(i) + (N-k+1)V(k)\Bigr]\\
	&=\frac{1}{N(k-\ell+2)}\Bigl[-k^2+(N+1)k-(N+1)\ell+\ell^2\Bigr].
	\end{align*}
	{\color{black} 
		The maximum value of $ F_+ $ is at $ \widehat{k}=\ell - 2+\sqrt{2N+6-4\ell}$, when $\ell<\frac{N+1}{2}$.
	}
	With similar arguments we get the expression for $k<\ell$. 
\end{proof}


Proposition \ref{prop:LineGraph} guarantees that there exists an optimal value of $ v^+ = k^{\star} $, which is 
placed on the left or on the right of $ v^- = \ell $ depending on the value of $ \ell $. This fact is quite intuitive: indeed, if $ v^- $ is not in the middle, the strategic agent $\piu$ is able to influence a larger amount of individuals by targeting an agent on the opposite side of $v^{-}$. What is surprising is that, in general, \textit{targeting an agent immediately next to $ v^- $ is not an optimal choice}. It is more effective to target an agent slightly on the opposite side, with some nodes of distance. This is probably due to the impact effect that also agent $ \meno$ would have, being close to $ v^+ $ too. Clearly, If $v^{-}$ is in the middle, then the optimal choice for $\piu$ is to cancel out its effect by targeting the same node. This is what happens in the popular game theoretical Hotelling model \cite{Hotelling:1929}.


\subsection{Tree Graphs}
\label{sec:TreeGraph}
For the Line Graph we have found an analytical solution for the STP, determining exactly the optimal position of $ v^+ $ in order to maximize the influence of $ \piu $. In an analogous way, we could think of extending the argument behind the previous section to a generic tree. Indeed, given the position of $ v^- $, for each possible choice of $ v^+ $ there exists just one path connecting $ v^- $ to $ v^+ $, thus leading to a similar situation as before. By considering the corresponding electrical network it is easy to see that each node that does not belong to this path is short-circuited with one on the path, i.e. the only voltage drops happen along this path.

On the other hand, when considering a generic tree, the computation of $ V(i) $ is not straightforward. Indeed, while for the Line Graph each intermediate node between $ v^- $ and $ v^+$ produces an identical voltage drop, for the Tree Graph each node belonging to the path between $ v^- $ and $ v^+ $ contributes to such drop proportionally to the number of nodes of its subtree (see Figure \ref{fig:treegraph_vl} for a better understanding).
\begin{figure}  
	\centering  \resizebox{.70\columnwidth}{!}{
	\begin{tikzpicture}
	[-,shorten >=1pt,auto,node distance=2cm,thick,main node/.style={circle,fill=blue!20},square/.style={regular polygon,regular polygon sides=4,inner sep=0.15em}]
	\tikzset{edge/.style = {<-,> = latex'}}
	
	\node[main node,square] (-) {$-$};
	\node[main node] (l) [below of=-]{$ v^- $};
	\node[main node] (a) [right of=l] {};
	\node (d) [right of=a] {$ \dots $};
	\node[main node] (b) [right of=d] {};
	\node[main node] (k) [right of=b] {$ v^+ $};
	\node[main node,square] (+) [above of=k] {$+$};
	
	\draw (l) -- (a) -- (d) -- (b) -- (k);
	\draw[edge] (l) -- (-);
	\draw[edge] (k) -- (+);
	
	\node at (2, -3) (2b) {...};
	\draw (a) --  (1.4,-3.3);
	\draw (a) -- (2.6,-3.3);
	
	\node at (6, -3) (2b) {...};
	\draw (b) --  (5.4,-3.3);
	\draw (b) -- (6.6,-3.3);
	
	\node at (-1, -1.9) (1b) {\vdots};
	\draw (l) --  (-1.3,-1.4);
	\draw (l) -- (-1.3,-2.6);
	
	\node at (9, -1.9) (kb) {\vdots};
	\draw (k) --  (9.3,-1.4);
	\draw (k) -- (9.3,-2.6);

	\end{tikzpicture}}
	\caption{Tree Graph with $ v^-$, $ v^+$}
	\label{fig:treegraph_vl}
\end{figure}
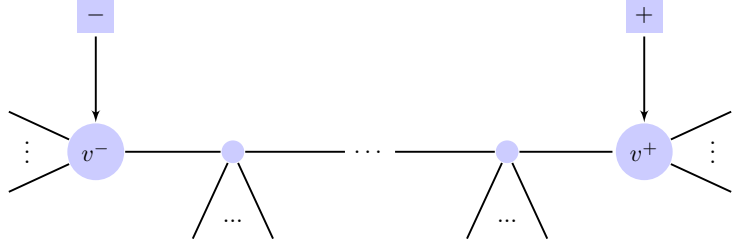
\noindent
Because of this complication, in order to compute $ F_+ $ for a generic Tree Graph, it becomes necessary to have more information about the tree.

More formally, let $ \mathcal{T=(I,E}) $ be a Tree Graph. Then, given a pair of distinct nodes $ i,j \in\mathcal{I}$, denote by $\mathcal{I}^{<ij}$ the subtree rooted at node $ i $ that does not contain node $ j $, along with the path from $ i $ to $ j $ (apart from $ i $), i.e.
$$ \mathcal{I}^{<ij}=\{h\in\mathcal{I}|\text{path from $ h $ to $ j $ goes through $ i $}\} $$
Similarly, denote by $c_i$ the cardinality of the subtree rooted at node $ i\in\{$path from $ v^- $ to $ v^+ \} $ made up by the nodes $ j\in \mathcal{I}^{<iv^-}\cap \mathcal{I}^{<iv^+} $. Using this formalism, the objective function $ F_+ $ can be written as follows:
\begin{align*}
F_+(k)&=\frac{1}{N}\Bigl[|\mathcal{I}^{<v^- v^+}|V(v^-)+|\mathcal{I}^{<v^+ v^- }|V(v^+)+\\
&+ \sum_{\substack{i\in\{\text{path from $v^- $ to $ v^+ \}$}}}  c_i V(i) \Bigr]
\end{align*}
It is clear from this expression that the optimal node lies on the path from $v^- $ to $v^+$ and therefore the objective function needs to be evaluated on these nodes only. Actually, we shall now show that the number of evaluations can be reduced further.

\begin{proposition}[Monotonicity over a branch]
	\label{th:prop1}
	Let $ \mathcal{T=(I,E)} $ be a tree. Let us consider $v^-=1$ as the root and consider a tree branch, a path going from the root to one of the leafs, denoting the nodes in the sequence by $i\in\{1,2,\ldots,L\}$. Then, there exists $k^{\star}\in\{1,\ldots,L\}$ such that the objective function $F_+(k)$ is monotone increasing in $k\in [1,k^{\star}]\cap\mathbb{N}$ and monotone decreasing in  $k\in[k^{\star},L]\cap\mathbb{N}$.
	
	\end{proposition}
\begin{proof}
From electrical analogy we obtain $\forall i=1,\dots,k$ that
	$V(i)={2i}/{(k+1)}-1$
and $V(j) = V(i),\ \forall j\in \mathcal{I}^{<i1}\cap \mathcal{I}^{<ik}$.
Then, noticing that $\mathcal{I}^{<i k}= \mathcal{I}^{<i L}, \forall k: i<k<L$, we have
 \begin{align*}
F_+(k)\!&=\!\frac{1}{N}\Bigl[\!|\mathcal{I}^{<1 k}|V(1)\!+\! \sum_{i=2}^{k-1}   |\mathcal{I}^{<i 1}\! \cap \mathcal{I}^{<i k}|V(i) \!+\!|\mathcal{I}^{<k1 }\!|V(k)\!\Bigr]\\
&=\!\frac{1}{N}\Bigl[\!|\mathcal{I}^{<1 L}\!|V(1)\!+\! \sum_{i=2}^{k-1} |\mathcal{I}^{<i 1}\! \cap \mathcal{I}^{<i L}\!| V(i) \!+\!|I^{<k1 }\!|V(k)\!\Bigr]
\\
&=\!\frac{1}{N}\Biggl[\sum_{i=1}^{k-1}c_i V({\color{black}{i}}) + V(k)\sum_{j=k}^{L}c_j\Biggr]\end{align*}
where $  c_i=|\mathcal{I}^{<i 1} \cap \mathcal{I}^{<i L}| $ and $\mathcal{I}^{<1L}=\mathcal{I}^{<11}\cap \mathcal{I}^{<1L}$.
Then, putting the expression for $V(i)$, we get
\begin{align*}
F_+(k)&=\frac{1}{N}\!\left[\sum_{i=1}^{k-1}c_i \left(\frac{2i}{k+1}-1\right)+ \left(\frac{2k}{k+1}-1\right)\sum_{j=k}^{L}c_j\right]\\
&=\frac{1}{N}\!\left[\sum_{i=1}^{k}c_i \left(\frac{2i}{k+1}-1\right)+ \left(\frac{2k}{k+1}-1\right)\sum_{j=k+1}^{L}c_j\right]\!.
\end{align*}
Notice that
\begin{align*}
&F_+(k+1)-F_+(k)\\
&\quad=\!\frac{1}{N}\!\left[\!\sum_{i=1}^{k}\!c_i\! \left(\frac{2i}{k+2}\!-\!\frac{2i}{k+1}\right)\!+\! \left(\frac{2k+2}{k+2}\!-\!\frac{2k}{k+1}\right)\!\!\sum_{j=k+1}^{L}\!c_j\!\right]\\
&\quad=\!\frac{1}{N}\!\left[\!\sum_{i=1}^{k}\!c_i\! \left(-\frac{2i}{(k+1)(k+2)}\right)\!+\! \frac{2}{(k+1)(k+2)}\sum_{j=k+1}^{L}\!c_j\!\right]\\
&\quad=\frac{2}{N(k+1)(k+2)} \left(\sum_{j=k+1}^{L}c_j - \sum_{i=1}^{k}i c_i\right)\\
\end{align*}
Let $G(k)$ be the quantity between brackets: observe that $G(k)$ is decreasing in $ k $ and that $$G(L-1)\leq c_L-\sum_{i=1}^{L-1}c_i=-(L-1)\leq 0.$$
We conclude that there exists $k^{\star}\in\{1,\ldots,L\}$ such that $F_+(k)$ is decreasing in $k\in\{k^{\star},\ldots,L\}$.
%
\end{proof}
Let us visit one node at a time starting from the root $ v^-=1 $. The strategic node $ \piu $, currently considering node $ k $, could move to one of its children, looking for a node that increases the objective function $ F_+ $. 
\begin{proposition}[Exploration of offspring population]
	\label{th:prop2}
	Let $ \mathcal{T=(I,E)} $ be a tree. 
	Let $ v^-\in\mathcal{I}$ be the root node, and consider the path from node $ v^-$ to a generic node $ k\in\mathcal{I} $. Let $ \mathcal{O}(k) $ be the offspring of $ k $, denoted as the set of nodes linked to $ k $ belonging to the subtree $ \mathcal{I}^{<kv^-} $.
	If there exists $ m \in \mathcal{O}(k) $ such that $ F(m)>F(k) $, then $F(n)<F(k)$ for each $ n\in \mathcal{O}(k)\setminus \{m\}$.
\end{proposition}

\begin{proof}
Let us represent $ \mathcal{T} $ as a pseudo-line branch, and let us denote the path from the root node $ v^-=1$ to $ k $, as the length-$ k $ path shown in Figure \ref{fig:treegraph_v1_3}. Let us assume that at least one node $ m $ in $ k $'s offspring $ \mathcal{O}(k) $ is such that $ F(m)>F(k) $, where $ |\mathcal{O}(k)|\geq 2 $, otherwise the proof would be trivial, and denote with $ n $ a generic node in such offspring different from $ m $.

\begin{figure} 
	\centering\resizebox{.75\columnwidth}{!}{

	\begin{tikzpicture}
	[-,shorten >=1pt,auto,node distance=2cm,thick,main node/.style={circle,fill=blue!20},square/.style={regular polygon,regular polygon sides=4,inner sep=0.15em}]
	\tikzset{edge/.style = {->,> = latex'}}	

	\node[main node] at (2,0) (k) {$ k $};
	\node[main node] at (4,1) (i) {$ m $};
	\node[main node] at (4,-1) (j) {$ n $};
	\node (d) [left of=k] {$ \dots $};
	\node[main node] (2) [left of=d] {2};
	\node[main node] (1) [left of=2]{1};
	\node[main node,square]  (-) [left of=1] {$ - $};
	
	\draw(1) -- (2) -- (d) -- (k);
	\draw (k) -- (i);
	\draw (k) -- (j);
	\draw[edge] (1) -- (-);
	
	\node at (4.7,1.1) (1b) {\vdots};
	\draw (i) --  (5.1,1.5);
	\draw (i) -- (5.1,0.5);
	
	\node at (4.7, -0.9) (1b) {\vdots};
	\draw (j) -- (5.1, -0.5);
	\draw (j) -- (5.1, -1.5);
	
	\node at (2, -1) (2b) {...};
	\draw (k) --  (1.4,-1.3);
	\draw (k) -- (2.6,-1.3);
	
	\node at (-2, -1) (2b) {...};
	\draw (2) --  (-1.4,-1.3);
	\draw (2) -- (-2.6,-1.3);
	
	\node at (-4, -1) (2b) {...};
	\draw (1) --  (-3.4,-1.3);
	\draw (1) -- (-4.6,-1.3);
	\end{tikzpicture}}
	\caption{Generic path from $ 1 $ to $ k $ on a tree with root $ v^-=1$. Once evaluated $F(k)$, Proposition \ref{th:prop2} considers to move to one of the children (here $m$ and $n$)}
	\label{fig:treegraph_v1_3}
\end{figure}
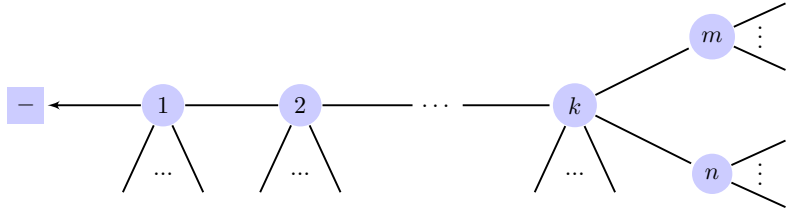
Let us consider the unique path from $v^-=1$ to $v^+$ and let $c_i^{(v^+)}$ be the cardinality of the subtree generating from each node in the line. 
It should be noticed that $ c_k^{(k)} = c_k^{(m)}+c_m^{(m)} = c_k^{(n)}+c_n^{(n)} $
and
$$ \xi=|\mathcal{I}^{<k1} \cap \mathcal{I}^{<km} \cap \mathcal{I}^{<kn}| =c_k^{(k)} - c_m^{(m)} - c_n^{(n)}$$
 By electrical analogy, we have
\begin{align*}
F_+(m) &=\frac{1}{N}\Biggl[ \sum_{i=1}^{k-1}c_i^{(i)} V_i^{(i)}(i)+c_k^{(m)} V_k^{(m)}(k)+c_m^{(m)} V_m^{(m)}(m)\Biggr]\\
&=\frac{1}{N}\Biggl[\sum_{i=1}^{k-1}c_i^{(i)} \left(\dfrac{2i}{k+2}-1\right)+\\
&+ (\xi+c^{(n)}_n) \left(\dfrac{2k}{k+2}-1\right) + c^{(m)}_m \left(\dfrac{2k+2}{k+2}-1\right)\Biggr]
\end{align*}
from which 
\begin{align*}
&F_+(m) - F_+(k)\\
&\quad =\frac{1}{N}\Biggl[\sum_{i=1}^{k-1}c^{(i)}_i \left(\dfrac{2i}{k+2}-\dfrac{2i}{k+1}\right)+\\
&\quad+ (\xi+c^{(n)}_n)\! \left(\dfrac{2k}{k+2}-\dfrac{2k}{k+1}\right)\! + c^{(m)}_m \!\left(\dfrac{2k+2}{k+2}-\dfrac{2k}{k+1}\right)\Biggr]\\
&\quad=\dfrac{2}{N(k+1)(k+2)}\left[-\sum_{i=1}^{k-1}i c^{(i)}_i - k(\xi+c^{(n)}_n)  + c^{(m)}_m\right]\\
&\quad=\dfrac{2}{N(k+1)(k+2)}[-f - k c^{(n)}_n  + c^{(m)}_m]
\end{align*}
where $f :=\sum_{i=1}^{k-1}i c^{(i)}_i + k\xi >0$.
By hypothesis $
F_+(m) - F_+(k) >0$, then
$c^{(m)}_m>k c^{(n)}_n +f>\frac{1}{k} c^{(n)}_n +f$.
We thus have
\begin{align*}
F_+(n) - F_+(k)& =
\dfrac{2}{N(k+1)(k+2)}(-f - k c^{(m)}_m  + c^{(n)}_n)\\
&\leq\dfrac{2}{N(k+1)(k+2)}[-(k+1)f]\leq0.
\end{align*}

\end{proof}
Proposition \ref{th:prop2} guarantees that at most one of its children can increase the objective function value. Then, it is useless to compute the objective function on the other nodes, if an improving node has already been found. 

Proposition \ref{th:prop1} and Proposition \ref{th:prop2} imply that, starting from the root $ v^- $ and moving $ v^+ $ from $ v^- $ to its first neighbors, only one of them will make $ F_+ $ increase. This is true also for such improving neighbor and it continuous, as going towards the leafs, until no improving neighbor is found, thereby identifying the optimal node. This leads to design the following algorithm in order to improve the maximum search algorithm.

%
%
\begin{algorithm}\caption{Tree Graph Single Targeting Algorithm [TGSTA]}
	
	\begin{algorithmic} \label{alg:TGSTA}
		\REQUIRE $\mathcal{T}=(\mathcal{I},\mathcal{E})$ tree graph, node $v^{-}$\\
		
		{\bf Initialization:}
		\STATE  $\quad$Root node $ r=v^{-} $ 
		\STATE $\quad$Number of visited nodes $s=0$
		\STATE $\quad$Evaluate $ F_+(r)$
		\STATE $\quad$Flag $ f=0 $
		\WHILE{$f = 0$}
		\STATE $ f=1 $
		\FOR{$ \ell\in \mathcal{O}(r) $}
		\STATE  $ s=s+1 $
		\STATE Evaluate $ F_+(\ell)$
		\IF{$ F_+(r) < F_+(\ell) $}
		\STATE $r = \ell , F_+(r) = F_+(\ell), f=0 $
		\STATE {\bf break for}
		\ENDIF
		\ENDFOR
		\ENDWHILE
		\RETURN $v^+=r$, $F_+(r)$, $ s $
	\end{algorithmic}
\end{algorithm}

\begin{theorem}[STP over trees]
	Let $\mathcal{T}=(\mathcal{I},\mathcal{E})$ be a tree, Algorithm \ref{alg:TGSTA} solves STP.
\end{theorem}
\begin{proof}
 Since $ F_+(\cdot) $ admits a maximum on each branch (see Proposition \ref{th:prop1}) and there is at most one initial node from which a monotonically increasing branch can start (see Proposition  \ref{th:prop2}), we conclude the convergence of the algorithm to the solution of STP.
\end{proof}

\begin{example2}
We consider 50 random trees: we start with a single individual in generation 0
 and, for each node, a number of children is generated according to a Poisson distribution with parameter $\lambda\in\{3,6,9,12\}.$ The strategic node $\meno$ is connected to a node chosen uniformly at random.
For each instance, the STP problem is solved using TGSTA (Alg.~4). Figure \ref{fig:Alg_tree1} depicts the 
average fraction of visited nodes as function of number of regular nodes in the network for different offspring distribution (different curves correspond to different value of $\lambda$). It should be noticed that TGSTA allows to reduce largely the computational complexity of maximum search. Moreover, when the size of the tree increases, then the gain becomes larger and just a small fraction of nodes needs to be explored.

\begin{figure}   [h]
	\begin{center}
	\includegraphics[trim=0cm 0cm 0cm 0cm,width=0.76\columnwidth]{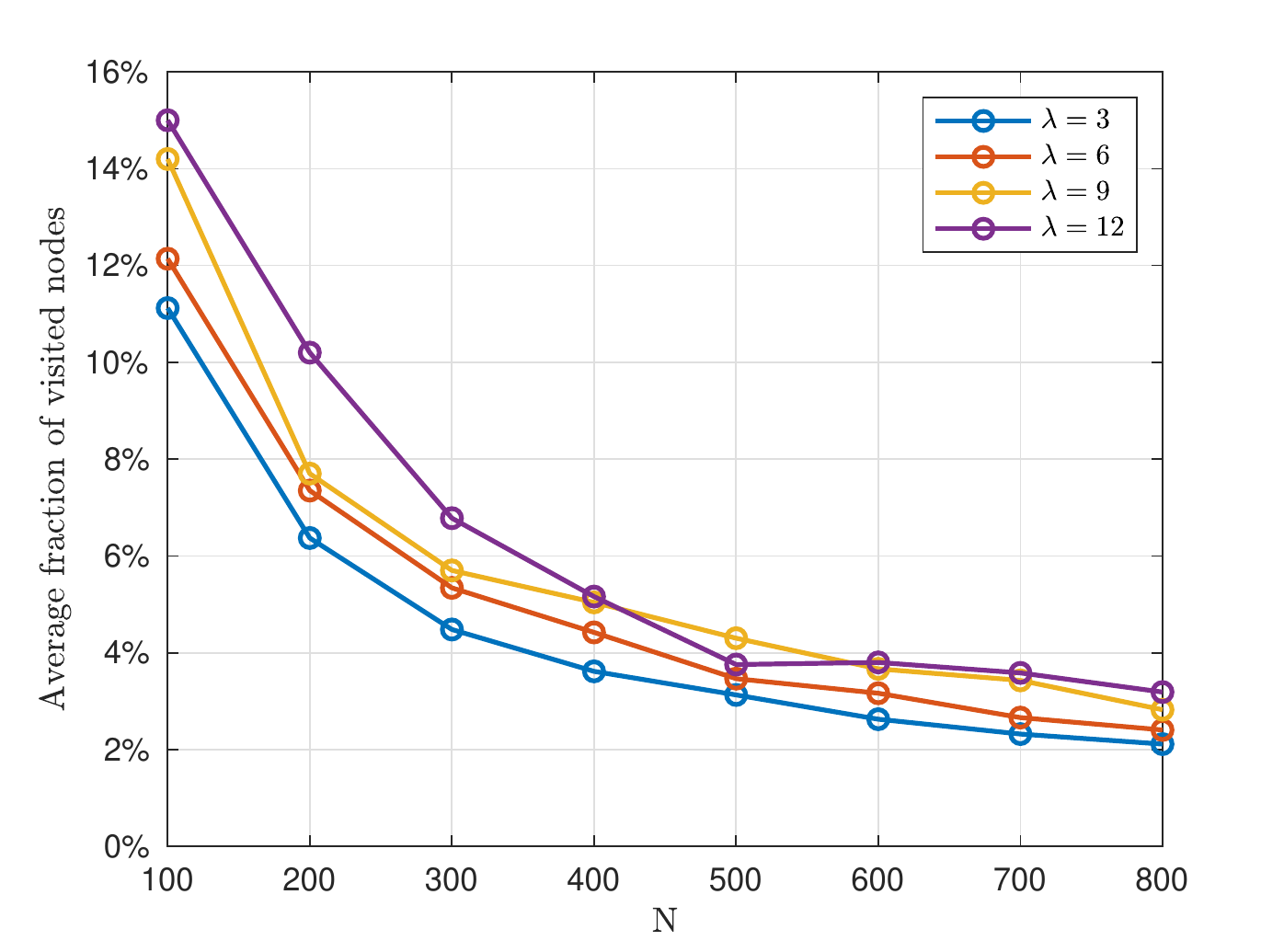}
	\end{center}
	\caption{TGSTA over random trees: Average fraction of visited nodes averaged over 50 experiments for different offspring distribution (Poisson$(\lambda)$).}\label{fig:Alg_tree1}
\end{figure}
\end{example2}

\section{Tree-like heuristics}\label{sect-tree-like}

We now apply the insights of the previous section and extend Algorithm \ref{alg:TGSTA} to graphs that are not trees.

\subsection{STP in Tree-like graphs}
We present now an algorithm, referred to \textit{Tree-like Single Targeting Algorithm}, that works as follows. When looking at the root's offspring, it does not stop looking at the first increasing $ F_+ $ value found, but it saves each \textit{improving node} (i.e. a node leading to an increasing value of $ F_+ $). In principle, we could use all of them as roots for the next iterations. Indeed, being the graph not a tree, it is possible to have more values in the nearest neighborhood leading to increasing values of $ F_+ $. Then, we decide to use only the node leading to the maximum improvement as the root for the next iteration. The code is summarized in Algorithm \ref{alg:TGSTA-mod-S}.

		\begin{algorithm}[h]\caption{Tree-like Single Targeting Algorithm (Tree-like-STA)}
		\begin{algorithmic} \label{alg:TGSTA-mod-S}
			\REQUIRE $\mathcal{G}=(\mathcal{V},\mathcal{E})$ graph, node $v^{-}$\\
			
			{\bf Initialization:}
			\STATE  $\quad$Root node $ r=v^{-} $ 
			\STATE $\quad$Number of visited nodes $s=0$
			\WHILE{$r \neq \emptyset$}
			\STATE $v=\emptyset$ empty set of improving nodes 
			\FOR{$ \ell\in \mathcal{O}(r) $}
			\STATE  $ s=s+1 $
			\STATE Evaluate $ F_+(\ell)$
			\IF{$ F_+(r) < F_+(\ell) $}
			\STATE $v = v\cup \{\ell\}$
			\ENDIF
			\ENDFOR
			\STATE $r = \argmax{\hat{v}\in v} F_+(\hat{v})$
			\ENDWHILE
			\RETURN $v^+=r$, $F_+(r)$, $ s $
		\end{algorithmic}
	\end{algorithm}

	\begin{example2}\label{ex:treelike_heur}

We now consider STP on 
 50 random generated Erdos-Renyi graphs with connectivity parameter $a\in\{1.5, 3,4.5, 6\}$, $ p=a\log(N)/N $. We solve the STP by using Tree-like-STA (Alg.~\ref{alg:TGSTA-mod-S}): denote by $F^{\star}$ the optimal value of the objective function and $\widehat{F}$ the value of function at node identified by Tree-like STA (Alg.~5). Figure \ref{fig:Alg_regular1} and Table \ref{table:Alg_regular1} show the number of visited nodes and the empirical probability of success as a function of the size of the network, where we declare a success when $|F_+^{\star}-\widehat{F}_+|/|F_+^{\star}|\leq1/\mathrm{e}$.

\begin{figure}
	\begin{center}
	\includegraphics[trim=0cm 0cm 0.2cm 0cm,width=0.76\columnwidth]{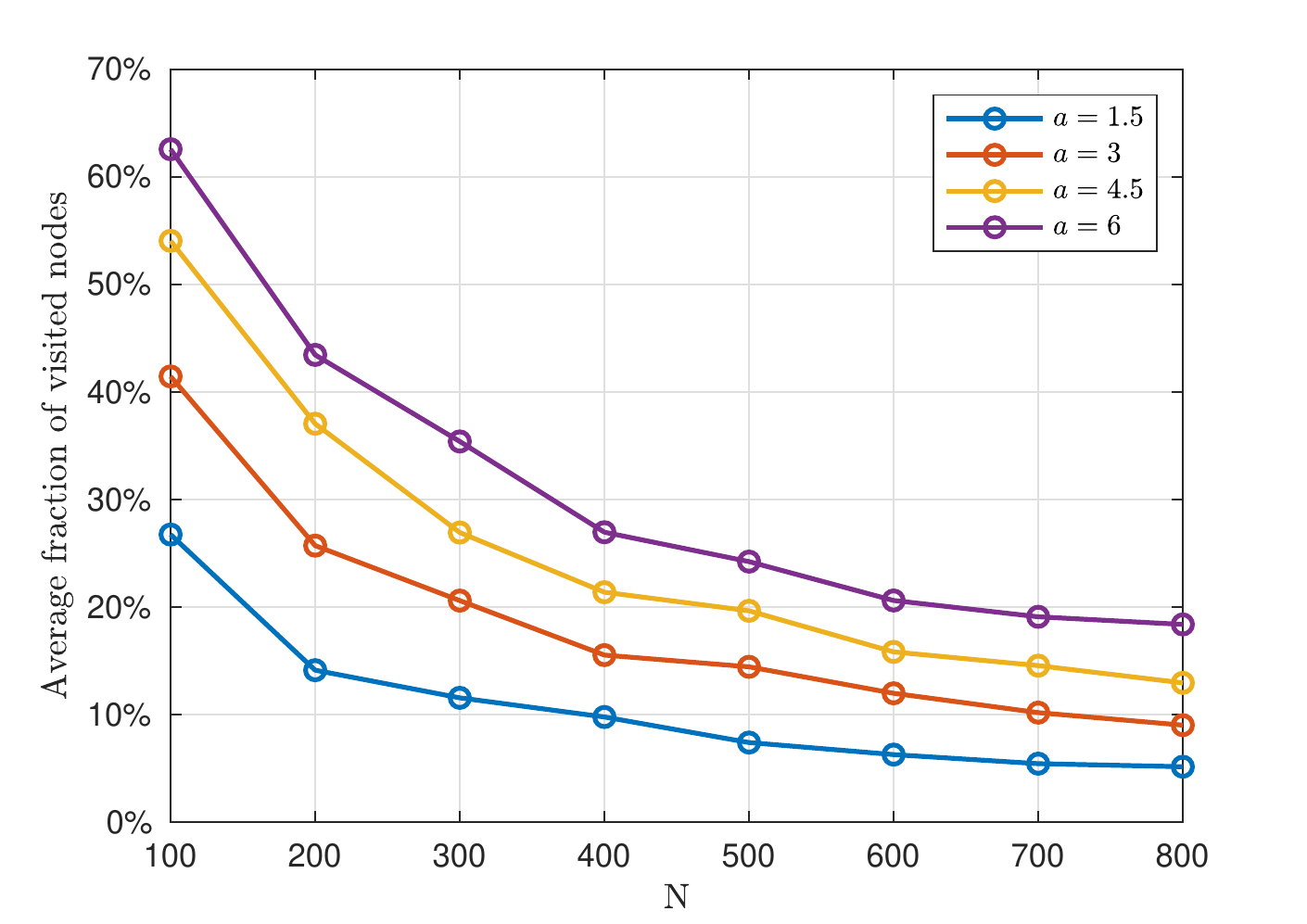}
	\end{center}
	\caption{Tree-like STA over Erdos-Renyi graphs: Average fraction of visited nodes averaged over 50 experiments for different connectivity parameter $ a $.}
	\label{fig:Alg_regular1}
\end{figure}

\begin{table}
\begin{center}
\begin{tabular}{c|cccc}
&$a=1.5$&$a=3$&$a=4.5$&$a=6$\\
\hline
$N=100$&0.900 & 0.960 & 0.980 & 1.000   \\
$N=200$&0.940 & 0.960 & 0.940 & 0.980   \\
$N=300$&0.840 & 0.940 & 0.960 & 0.960   \\
$N=400$&0.840 & 0.920 & 0.880 & 0.940   \\
$N=500$&0.840 & 0.960 & 0.960 & 0.940   \\
$N=600$&0.800 & 0.900 & 0.940 & 0.980   \\
$N=700$&0.920 & 0.940 & 0.880 & 0.960   \\
$N=800$&0.860 & 0.860 & 0.900 & 0.880
\end{tabular}
\end{center}
\caption{Tree-like STA over Erdos-Renyi graphs: Empirical probability of success computed on 50 experiments for different connectivity parameter $ a $.}
\label{table:Alg_regular1}
\end{table}
Some remarks are in order. The number of visited nodes increases when the graph is less sparse and the probability of success is larger than 0.8 for all networks. Additionally, as the exploration decreases, accuracy gets worse.
\end{example2}


\begin{example2}
We now test Algorithm \ref{alg:TGSTA-mod-S} on a real large-scale online social network: the Facebook ego-network, retrieved from Stanford Large Network Dataset Collection (\url{https://snap.stanford.edu/data/ egonets-Facebook.html}). 
This dataset contains anonymized personal networks of connections between friends and the size of the graph associated is $ |\mathcal{V}| = 4039 $, while the number of links is equal to $|\mathcal{E}| =  88234 $. Such graph is extremely sparse, since the number of nonzero elements $|\mathcal{E}|/|\mathcal{V}|^2 \approx 5*10^{-3} $.
10 instances of STP are generated by linking agent $\meno$ to a random regular agent. 
We find that Algorithm \ref{alg:TGSTA-mod-S} reaches the optimum (solving STP by means of the brute-force approach). We obtain that the average fraction of visited nodes is equal to $ 30\% $. 
\end{example2}

When the number of links placed by strategic agent $ \meno $ is greater than 1, we use a generalized version of the tree-like heuristic: among the nodes linked to $ \meno $, we select as the root $ v^- $ from which Algorithm \ref{alg:TGSTA} is started the one with smallest degree. The reasoning behind this algorithm, supported by empirical simulations, is that in sparse graphs it is easier to move away from not relevant nodes rather than vice versa. Indeed, if starting the algorithm from high degree nodes, the first steps would generally be more affected by the noise produced by the strong influence of $ \meno $.

%

\subsection{OTP on Tree-like graphs}

	For the general OTP, when both strategic agents have a number of available or placed nodes greater than 1, we propose an algorithm that is a generalized version of previous single targeting algorithms over tree-like graphs. Specifically, it simulates a greedy heuristic where a sub-optimum is found at each step. Specifically, this is done by selecting at each step a different root node among the ones linked to $ \meno $.

%
%

%
\begin{example2}
	Let us now compare the Tree-like Heuristics (Alg.~5) with Greedy Heuristics (Alg.~3) on random generated Erdos-Renyi graphs of parameters $ N = 200$ and $ p=0.1 $. We generate $ 15 $ random graphs and we link $ |A^-|= 3 $ nodes randomly selected to the strategic agent $ - $. The results are reported below:
\begin{table} [H]
	\centering
	\resizebox{0.675\columnwidth}{!}{%
		\begin{tabular}{c|cc}
			             & Average $ F_+ $        & Average fraction of visited \\
			             &									&nodes at each step \\
			\hline
			Tree-like Heuristic            & 0.2555           & $27\%$ \\
			Greedy Heuristic             & 0.2566          & $100\%$ \\
		\end{tabular}}
\end{table}
We can easily see how the $ F_+ $ values are really close to each other, while the average number of computations is almost reduced by one quarter by the Tree-like Heuristic. 
\end{example2}

\section{Summary and concluding remarks}
\label{sec:conclusions}

In this paper we have considered the optimal targeting problem in a social network where two strategic agents are competing. We have both studied the problem for special classes of graphs and proposed heuristics for general graphs. Available heuristics essentially rely on two approaches, namely, the identification of ``important'' nodes by some general-purpose measure of centrality --the simplest such {\em centrality measure} is just the node degree--, or a {\em greedy} approach, in which nodes are targeted one-at-the-time and which is justified by the submodularity of the objective function. We have exemplified these two approches by Algorithms 1 and 2. 

Starting from this background, we have studied the problem in complete graphs and on trees: in the former case, we have explicitly found the optimal solution; in the latter, we have identified two key properties that greatly simplify its computation. Complete and tree graphs are extremal examples, as they respectively feature maximal and minimal connectivity (only one path connects any two nodes on a tree). The insights from these two examples are more broadly relevant, as they translate into novel heuristic approaches to the optimal targeting problem.
The complete graph suggests the approach of {\em blocking} nodes that have been targeted by the adversary: this approach has zero cost, meaning that it requires no evaluations of the equilibrium opinions (Alg.~3). The tree graphs suggest the approach of {\em tree-like exploration}, which allows greatly reducing the cost of the greedy approach (Alg.~5).

These four heuristic approaches can --and should-- be combined in designing heuristic algorithms for the optimal targeting problem. Note, for instance, that our Algorithm~3 does combine blocking and greedy approach.
The most suitable combination shall depend on the known properties of the underlying graph and its choice will require to address the trade-off between cost and accuracy.
Let us for instance consider the choice of whether to block or not the opponent's influence by targeting the same nodes. This is typically a wise choice, at least if one has the possibility of targeting more nodes than its opponent. However, a blocking approach is less effective if the graph is very sparse or if the opponent is linked to marginal nodes: the latter issue can be addressed by restricting the blocking procedure to the high degree nodes linked to $ \meno $. More generally, if the graph is clearly split between high and low degree nodes, a degree-based approach would be a good choice. Also the choice of applying the tree-like approach to accelerate the greedy algorithm should mainly be based on the graph structure. Indeed, if the graph is locally tree-like or sparse, a tree-like approach would be well-performing, improving the complexity of the heuristic.  

Further research should concentrate on refining these guidelines for heuristics.

\bibliographystyle{IEEEtran}
\bibliography{Biblio_optimal_targeting}

\end{document}